\pdfoutput=1
\documentclass{article}

\usepackage[lmargin=1.23in,rmargin=1.23in,bmargin=1.1in,tmargin=1.1in]
{geometry}

\usepackage{amsthm}
\usepackage{amssymb}
\usepackage{amsmath}
\usepackage[pdftex]{graphicx}
\usepackage{xspace}
\usepackage{color}
\usepackage{bold-extra}
\usepackage[T1]{fontenc}
\usepackage[sort]{cite}
\usepackage{todonotes}

\newtheorem{theorem}{Theorem}

\newtheorem{lemma}{Lemma}

\newcommand{\congest}{\ensuremath{\mathcal{CONGEST}}\xspace}

\newcommand{\whp}{w.h.p.\@\xspace}

\newcommand{\fontmathtext}[1]{\mathsf{#1}}%
\DeclareRobustCommand{\alorithm}[1]{{\ensuremath{{\cal A}_\fontmathtext{#1}}}\xspace}
\DeclareRobustCommand{\AlgoHC}{\alorithm{HC}}

\DeclareGraphicsExtensions{.pdf,.jpeg,.png}
\usepackage[caption=false,font=normalsize,labelfont=sf,textfont=sf]{subfig}
\usepackage{url}

\title{A Distributed Algorithm for Finding Hamiltonian
  Cycles in Random Graphs in $O(\log n)$ Time} 

\author{Volker Turau\\
 	Hamburg University of Technology, Institute of Telematics\\Am Schwarzenberg-Campus 3, 21073 Hamburg, Germany\\ turau@tuhh.de}
 \date{}

\begin{document}

\maketitle

\begin{abstract}
  It is known for some time that a random graph $G(n,p)$ contains \whp
  a Hamiltonian cycle if $p$ is larger than the critical value
  $p_{crit}= (\log n + \log \log n + \omega_n)/n$. The determination
  of a concrete Hamiltonian cycle is even for values much larger than
  $p_{crit}$ a nontrivial task. In this paper we consider random
  graphs $G(n,p)$ with $p$ in $\tilde{\Omega}(1/\sqrt{n})$, where
  $\tilde{\Omega}$ hides poly-logarithmic factors in $n$. For this
  range of $p$ we present a distributed algorithm \AlgoHC that finds
  \whp a Hamiltonian cycle in $O(\log n)$ rounds. The algorithm works
  in the synchronous model and uses messages of size $O(\log n)$ and
  $O(\log n)$ memory per node.
\end{abstract}






\setcounter{page}{1}
\section{Introduction}

Surprisingly few distributed algorithms have been designed and
analyzed for random graphs. To the best of our knowledge the only work
dedicated to the analysis of distributed algorithms for random graphs
is \cite{Levy:05,Krzywdzinski:15,Pandurangan:2018}. This is rather
surprising considering the profound knowledge about the structure of
random graphs available since decades \cite{Boll:01,Frieze:15}. While
algorithms designed for general graphs obviously can be used for
random graphs the specific structure of random graphs often allows to
prove asymptotic bounds that are far better. In the classical
Erd\H{o}s and R\'{e}nyi model for random graphs a graph $G(n,p)$ is an
undirected graph with $n$ nodes where each edge independently exists
with probability $p$ \cite{Erdos:1959}. The complexity of algorithms
for random graphs often depends on $p$, e.g., Krzywdzi{\'n}ski et al.\
\cite{Krzywdzinski:15} proposed a distributed algorithm that finds
\whp a coloring of $G(n,p)$ with $18np$ colors in $O(\ln\ln p^{-1})$
rounds.

In this work we focus on finding Hamiltonian cycles in random graphs.
The decision problem, whether a graph contains a Hamiltonian cycle, is
NP-complete. It is a non-local graph problem, i.e., it is required to
always consider the entire graph in order to solve the problem. It is
impossible to solve it in the {\em local neighborhoods}. For this
reason there is almost no work on distributed algorithms for finding
Hamiltonian cycles in general graphs. On the other hand it is well
known that $G(n,p)$ contains \whp a Hamiltonian cycle, provided $p \ge
p_{crit}=(\log n + \log \log n + \omega(n))/n$, where $\omega(n)$
satisfies $\lim_{n\to\infty}\omega(n)=\infty$ \cite[Th.~8.9]{Boll:01}.
There is a large body of work on sequential algorithms for computing
\whp a Hamiltonian cycle in a random graph (e.g.
\cite{posa:76,Angluin:1979,Shamir:1983,Boll:87,Thomason:89}).

We are only aware of two distributed algorithms for computing
Hamiltonian cycles in random graphs. The algorithm by Levy et al.\
\cite{Levy:05} outputs \whp a Hamiltonian cycle provided
$p =\omega(\sqrt{\log n}/n^{1/4})$. This algorithm works in
synchronous distributed systems, terminates in linear worst-case
number of rounds, requires $O(n^{3/4 + \epsilon})$ rounds on
expectation, and uses $O(n)$ space per node. The algorithm of
Chatterjee et al. \cite{Pandurangan:2018} works for
$p \ge c \log n/n^\delta$ ($0<\delta \le 1$) and has a run time of
$\tilde{O}(n^\delta)$.

The search for a distributed algorithm for a Hamiltonian cycle is
motivated by the usage of virtual rings for routing in wireless
networks \cite{malkhi:2009,turau:2017}. A virtual ring is a directed
closed path involving each node of the graph, possibly several times.
Virtual rings enable routing with constant space routing tables,
messages are simply forwarded along the ring. The downside is that
they may incur a linear path stretch. To attenuate this, distributed
algorithms for finding {\em short} virtual rings have been proposed
\cite{helaryVirtualRingDFT:1987,turau:2017}. Hamiltonian cycles are
the {\em shortest} possible virtual rings and therefore of great
interest. Short virtual rings are also of interest for all token
circulation techniques as discussed in \cite{Franceschelli:11}. Kim et
al. discuss the application of random Hamiltonian cycles for
peer-to-peer streaming \cite{Kim:12}. Rabbat et al.\ present
distributed optimization algorithms for in-network data processing,
aimed at reducing the amount of energy and bandwidth used for
communication based on Hamiltonian cycles \cite{Rabbat:05}, see also
\cite{Sommer:08}.

This paper uses the synchronous \congest model, i.e., each message
contains at most $O(\log n)$ bits. Furthermore, each node has only
$O(\log n)$ bits of local memory. Without these two assumptions there
is a very simple solution provided the nodes have unique identifiers.
First a BFS-tree rooted in a node $v_0$ is constructed. Then the
adjacency list of each node is convergecasted to $v_0$ which applies a
sequential algorithm to compute \whp a Hamiltonian path (see
Sec.~\ref{sec:stateArt}). The result is broadcasted into the graph and
thus each node knows its neighbor in the Hamiltonian cycle. This can
be achieved in $O(diam(G))$ rounds. Note that if $p=\omega_n \log n/n$
then \whp $diam(G(n,p))= O\left(\frac{\log n}{\log np}\right)$
\cite{Chung:01,Frieze:15}. In particular for $p$ in
$\tilde{\Omega}(1/\sqrt{n})$ \whp the diameter of $G(n,p)$ is constant
\cite{Boll:81}.

For the stated restrictions on message size and local storage we
propose an algorithm that terminates in a logarithmic number of rounds,
this is a significant improvement over previous work
\cite{Levy:05,Pandurangan:2018}. Our contribution is the distributed
algorithm \AlgoHC, its properties can be summarized as follows.


\begin{theorem}\label{thm:main-contribution}
  Let $G(n,p)$ with $p \ge (\log n)^{3/2}/\sqrt{n}$ be a random graph.
  Algorithm \AlgoHC computes in the synchronous model \whp a
  Hamiltonian cycle for $G$ using messages of size $O(\log n)$.
  \AlgoHC terminates in $O(\log n)$ rounds and uses $O(\log n)$
  memory per node. 
\end{theorem}



\subsection{Related Work}\label{sec:stateArt}
P{\'o}sa showed already in 1976 that almost all random graphs with $cn
\log n$ edges possess a Hamiltonian cycle \cite{posa:76}. Later
Koml{\'o}s et al.\ determined the precise threshold $p_{crit}$ for the
existence of a Hamiltonian cycle in a random graph \cite{Komlos:1983}.
A sequential deterministic algorithm that works \whp at this threshold
requiring $O(n^{3+o(1)})$ time is due to Bollob{\'a}s et al.\
\cite{Boll:87}. For larger values of $p$ or restrictions on the
minimal node degree, more efficient algorithms are known
\cite{Angluin:1979,Frieze:2015}. The algorithm of Thomason finds a
Hamiltonian path or shows that no such path exists provided $p\ge 12
n^{-1/3}$ \cite{Thomason:89}.

The above cited algorithms were all designed for the sequential
computing model. Some exact algorithms for finding Hamiltonian cycles
in $G(n,p)$ on parallel computers have been proposed
\cite{Frieze:1987}. The first operates in the EREW-PRAM model and uses
$O(n\log n)$ processors and $O(\log^2 n)$ time, while the second one
uses $O(n\log^2 n)$ processors and $O((\log\log n)^2)$ time in the
P-RAM model. MacKenzie and Stout proposed an algorithm for CRCW-PRAM
machines that operates in $O(\log^\ast n)$ expected time and requires
$n/\log^\ast n$ processors \cite{MacKenzie:1993}. Apart from the above
mentioned work \cite{Levy:05,Pandurangan:2018} we are not aware of any
other distributed algorithm for this problem.




There are several approaches to construct a Hamiltonian cycle. The
approach used by Levy et al.\ at least goes back to the work of
MacKenzie and Stout \cite{MacKenzie:1993}. They initially construct a
small cycle with $\Theta(\sqrt{n})$ nodes. As many as possible of the
remaining nodes are assorted in parallel into $\sqrt{n}$
vertex-disjoint paths. During the final phase, each path and each
non-covered vertex is patched into the initial cycle.

The second approach is used in the proofs to establish the critical
value $p_{crit}$ (e.g., \cite{posa:76,Krivelevich:2016}) and all
derived sequential algorithms (e.g., \cite{Boll:87}). Initially a
preferably long path is constructed, e.g., using a depth first search
algorithm \cite{Frieze:2015}. This path is extended as long as the
node at the head of the path has a neighbor that is not yet on the
path. Then the path is {\em rotated} until it can be extended again. A
rotation of the path cuts off a subpath beginning at the head,
reverses the order of the subpath's nodes, and reattaches the subpath
again. The procedure stops when no sequence of rotations leads to an
extendable path. The algorithm in \cite{Pandurangan:2018} follows this
approach.






\section{Computational Model and Assumptions}
This work employs the synchronous \congest model of the
\emph{distributed message passing model} \cite{Peleg:2000}, i.e., each
message contains at most $O(\log n)$ bits. Furthermore, each node has
only $O(\log n)$ bits of local memory. The communication network is
represented by an undirected graph $G=(V,E)$, where $V$ is a set of
$n$ processors (nodes) and $E$ represents the set of $m$ bidirectional
communication links (edges) between them. Each node carries a unique identifier. Communication between nodes
is performed in synchronous rounds using messages exchanged over the
links. Upon reception of a message, a node performs local computations
and possibly sends messages to its neighbors. These operations are
assumed to take negligible time.

The prerequisite of Algorithm \AlgoHC is a distinguished node $v_0$
which is the starting point of the Hamiltonian cycle and acts as a
coordinator in the final phases of \AlgoHC. The results proved in this
work hold {\em with high probability} (\whp) which means with
probability tending to $1$ as $n \rightarrow \infty$. The
probabilities $p$ considered in this paper always depended on $n$,
e.g., $p \ge (\log n)^{3/2}/\sqrt{n}$, and we always assume that $\lim_{n \rightarrow
  \infty} p=0$. 



\section{Informal Description of Algorithm \AlgoHC}
\label{sec:overview}

Algorithm \AlgoHC operates in sequential phases, each of them succeeds
\whp The first two phases
last $O(\log n)$ rounds. Each subsequent phase requires a constant
number of rounds only. Phase 0 lasts $3(3\log n -1)$ rounds and
constructs a path $P$ of length $3\log n$ starting in $v_0$. In the
next $3\log n $ rounds Phase 1 closes $P$ into a cycle $C$ of length at
most $4\log n$. The following $16 \log(n)$ phases are called the
middle phases. In each of those phases the number of nodes in $C$
is increased. The increase is by a constant factor until $C$ has $n/7$
nodes. Afterwards, the increase declines roughly linearly until $C$ has
$n-3\log n$ nodes. In each middle phase the algorithm tries to
concurrently integrate as many nodes into $C$ as possible. This is
achieved by replacing edges ($v,w)$ of $C$ by two edges $(v,x)$ and
$(x,w)$, where $x$ is a node outside of $C$. At the end of the middle
phases \whp $C$ has more than $n-3\log n$ nodes.

The integration of the remaining $3\log n$ nodes requires a more
sophisticated algorithm. This is done in the final phases. The idea is
to remove two edges -- not necessarily adjacent -- from $C$ and insert
three new edges. This requires to reverse the edges of a particular
segment of $C$ of arbitrary length. Thus, this is no longer a local
operation. Furthermore, segments may overlap and hence, the
integration of several nodes can only be performed sequentially. Thus,
this task requires coordination. Node $v_0$ takes over the role of a
coordinator.

At the beginning of each final phase all nodes outside $C$ that can be
integrated report this to $v_0$, which in turn selects one of these
nodes to perform this step. For this purpose a tree routing structure
is set up, so that each node can reach $v_0$ \whp in $3$ hops. In
order for the nodes of the segment to perform the reordering
concurrently, the nodes of $C$ are numbered in an increasing order
(not necessarily consecutively) beginning with $v_0$. The assignment
of numbers is embedded into the preceding phases with no additional
overhead. The numbering is also maintained in the final $3\log n$
integration steps. In order to accomplish the integration in a
constant number of rounds -- i.e., independent of the length of the
segment -- node $v_0$ floods the numbers of the terminal nodes of the
segment to be reversed into the network. Upon receiving this
information, each node can determine if it belongs to the segment to
be reversed and can recompute its number to maintain the ordering.
Note that this routing structure requires only $O(\log n)$ memory per
node. Each of the $3\log n$ final phases lasts a constant number of
rounds.

Algorithm \AlgoHC stops when either $C$ is a Hamiltonian cycle or no
more nodes can be integrated into $C$. The first event occurs \whp


\section{Formal Description}
Algorithm \AlgoHC operates in synchronous rounds. By counting the
rounds a node is always aware in which round and therefore also in
which phase it is. Each phase lasts a known fixed number of rounds. If
the work is completed earlier, the network is idle for the remaining
rounds. This requires each node to know $n$. Algorithm \AlgoHC
gradually builds an oriented cycle $C$ starting with node $v_0$. The
cycle is maintained as a doubly linked list to support insertions. The
orientation of $C$ is administered with the help of variable $next$ --
initially $null$ -- which stores the identifier of the next node on
the cycle in clockwise order. Also, $v.next= null$ indicates that $v$
is not yet on the cycle. In the following each phase is described in
detail.



\subsection{Pre-processing}
\label{sec:preph}
The algorithm is started by node $v_0$ which executes algorithm
\textsc{Flood} \cite{Peleg:2000} to construct a BFS tree. By
Lemma~\ref{diameter} the diameter of $G$ is \whp  at most 3. Thus, in
$3$ rounds a BFS tree rooted in $v_0$ is constructed (Lemma 5.3.1,
\cite{Peleg:2000}). After a further $6$ rounds each node is aware of
$n$ the number of nodes in the network. This allows to run each phase
for the stated number of rounds.

\subsection{Phase 0}
\label{sec:ph0}

In phase 0 an oriented path $P$ starting in $v_0$ of length $3\log n$
is constructed. Phase 0 lasts $3(3\log n - 1)$ rounds. Initially
$P=\{v_0\}$ and $v_0.next=v_0$. The following steps are repeated $3\log
n - 1$ times.
\begin{enumerate}
\item The final node $v$ of $P$ sends an invitation message to all
  neighbors. All neighbors not on $P$ (i.e., nodes with $next=null$)
  respond to $v$.
\item If $v$ does not receive any response the algorithm halts.
  Otherwise $v$ randomly selects among the nodes that have responded a
  node $w$, sets $v.next := w$, informs $w$ that it is the new final
  node, and instructs $w$ to continue with phase 0. This message
  includes the id of node $v_0$, i.e., at any point in time all nodes
  of $P$ know $v_0$.
  
\end{enumerate}

\subsection{Phase 1}
\label{sec:ph1}
In phase 1 the path $P$ is extended into an oriented cycle $C$ of length
at most $4\log n$. The following steps are repeated at
most $\log n$ times. Phase 1 lasts $3\log n$ rounds.
\begin{enumerate}
\item The final node $v$ of $P$ sends an invitation message to all
  neighbors, the message contains the id of node $v_0$. All
  neighbors not on $P$ respond to $v$. The response includes the
  information whether the recipient is connected to $v_0$.
\item If $v$ does not receive any response the algorithm halts. If at
  least one responding node is connected to $v_0$, then $v$
  randomly selects such a node $w$, sets $v.next=w$, and informs $w$
  to close the cycle $C$, i.e., to set $w.next=v_0$. Otherwise $v$
  randomly selects a responding node $w$ to extend $P$ as in phase 0
  and instructs $w$ to repeat phase 1.
\item If after $\log n$ repetitions $P$ is not a cycle then the
  algorithm halts otherwise the middle phases start.
\end{enumerate}

\subsection{Middle Phases}
\label{sec:ph2}
While in the first two phases actions were executed sequentially, in
the middle phases many nodes are integrated concurrently. In each of
the subsequent phases the following steps are performed (see
Fig.~\ref{fig:mph} for an example). Each of the $16 \log n$ middle
phases is performed in three rounds.
\begin{enumerate}
\item Each node $w$ on $C$ broadcasts its own id and the id of its
  predecessor on $C$ using message $I_1$.
\item If a node $v$ outside $C$ receives a message $I_1$ from a node
  $w$ such that the predecessor of $w$ on $C$ is a neighbor of $v$, it
  inserts $w$ into the set $C_v$.
\item Each node $v$ outside $C$ with $C_v\not= \emptyset$ randomly
  selects a node $w$ from $C_v$ and sends an invitation message $I_2$
  to the  predecessor of $w$ on $C$.
\item Each node $w\in C$ that received an invitation $I_2$
  randomly selects a node $v$ from which it received an invitation,
  sets $w.next = v$, and informs $v$ with acceptance message $I_3$ to
  set its variable $next$ to the old successor $w'$ of $w$. In other
  words the edge $(w,w')$ is replaced by the edges $(w,v)$ and
  $(v,w')$.
\end{enumerate}
Individual extensions do not interfere with each other. Each node
outside $C$ gets in the last round of a middle phase at most one
request for extension and for each edge of $C$ at most one request is
sent.

\begin{figure}[h]
  \centering
\includegraphics[scale=1.0]{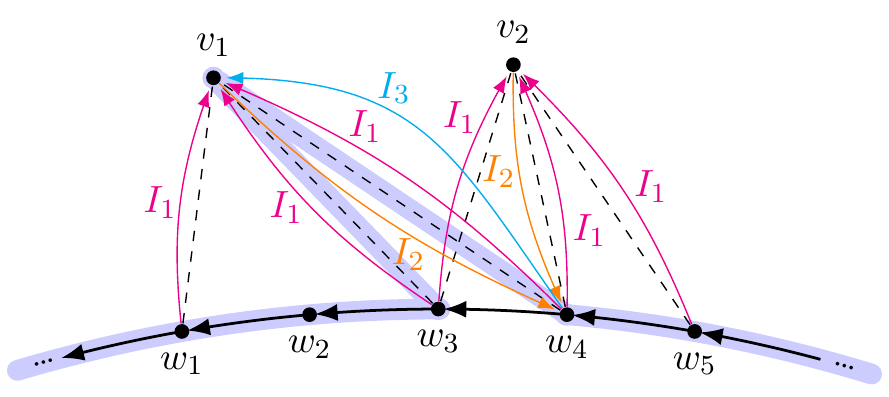}  
\caption{The integration of nodes during a middle phase:
  Nodes $w_i$ sent a message $I_1$ to all nodes outside $C$ (red
  arrows). Nodes $v_1$ and $v_2$ sent a message $I_2$ back to $w_4$;
   $v_2$ might have also selected $w_4$ and sent $I_3$ to
  $w_5$. Node $w_4$ selected $v_1$  and sent back
  message $I_3$. Edge $(w_4,w_3)$ is replaced by the edges $(w_4,v_1)$
  and $(v_1,w_3)$. The extended cycle is depicted by the blue ribbon.}
  \label{fig:mph}
\end{figure}



\subsection{Final Phases}
\label{sec:phl}

After the completion of the middle phases the cycle $C$ has \whp at
least $n-3\log n$ nodes. At that point the expected number of nodes
$v\in V\setminus C$ that send an invitation $I_2$ becomes too low to
complete the cycle. Therefore, the integration of the remaining
nodes requires a more complex integration procedure as depicted in
Fig.~\ref{fig:fph}. The procedure of the final phases is as follows.
Each node $v\in V\setminus C$ with identifier $id$ sends a message
$I_1(id)$ to each of its neighbors. A node $w_1\in C$ that receives a
message $I_1(id)$ sends a message $I_2(id)$ to its neighbor $w_2$ on
$C$ in clockwise order. If $w_2$ also received a message $I_1(id)$
(with the same id), then nodes $w_1,w_2$ and the initiating
node $v$ with identifier $id$ form a triangle. Then $v$ can be
directly integrated into $C$ as done in the middle phases. In this
case $w_1$ asks $v$ to initiate the integration step.

\begin{figure}[h]
  \centering
\includegraphics[scale=1.0]{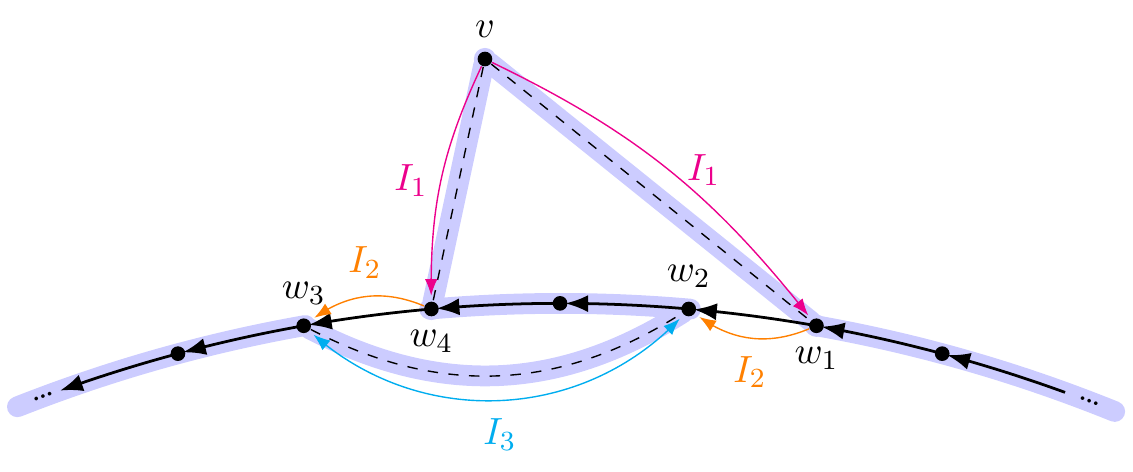}  
\caption{The integration of node $v$ into $C$ during the final phase.
  The thin red arrows indicate the flow of the messages $I_1,I_2$, and
  $I_3$ initiated by $v$. The extended cycle is depicted by the blue
  ribbon. The edges $(w_1,w_2)$ and $(w_4,w_3)$ are replaced by the
  edges $(w_1,v)$, $(v,w_4)$, and $(w_2,w_3)$. The order of the edges
  between $w_4$ and $w_2$ is reversed.}
  \label{fig:fph}
\end{figure}

Otherwise, if node $w_2$ did not receive a message $I_1(id)$, then it
sends a message $I_3(id)$ to all neighbors that are on $C$. If a node
$w_3$ on $C$ that receives this message $I_3(id)$ also received a
message $I_2(id)$ from its predecessor $w_4$ on $C$, then node $v$ can
be integrated into $C$ as shown in Fig.~\ref{fig:fph}. This is
achieved by replacing edges $(w_1,w_2)$ and $(w_4,w_3)$ from $C$ by
edges $(w_1,v)$, $(v,w_4)$, and $(w_2,w_3)$. Also, the edges on the
segment from $w_2$ to $w_4$ must be traversed in opposite order, note
that the number of nodes between $w_2$ and $w_4$ is not bounded. A
naive explicit reversing of the order of the edges on the middle
segment may require more than $O(\log n)$ rounds. Thus, we propose a
different approach.

Apart from the reversal of the edges in the middle segment this
integration can be implemented within five rounds. Node $w_3$ informs
$v$ about this integration possibility, this notification also
includes the identifiers of nodes $w_4$ and $w_2$. Furthermore, the
participating nodes $w_4,w_2$ and $w_1$ are also informed. The
approach to invert the middle segment in a constant number of rounds
is explained below.


Unfortunately there is another issue. While each node outside $C$ can
be integrated individually, these integration steps cannot be executed
concurrently. A problem arises if the segments, which are inverted
(e.g.\ from $w_2$ to $w_4$), overlap. This can result in two separate
cycles as shown in Fig.~\ref{fig:twoCycles}. Even if the integration
of the remaining nodes is performed sequentially, a problem appears if
the reversal of the middle segment is not made explicit. In this case
the nodes that receive an $I_1$ message may not have a consistent view
with respect to the clockwise order of $C$.

\begin{figure}[h]
  \centering
\includegraphics[scale=1.0]{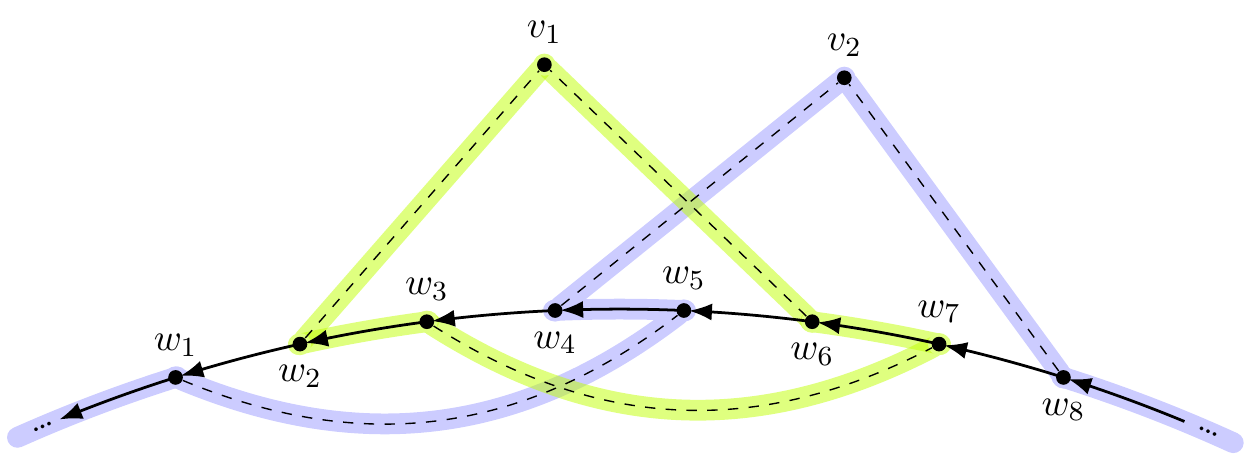}  
\caption{The depicted scenario shows that the integration of two nodes with
  overlapping segments cannot be performed concurrently as this would
  lead to two cycles (shown in green and blue). If $v_2$ would be
  integrated first, then $v_1$ can no longer be integrated, since the
  predecessor  of $w_6$ is then $w_5$ which is not connected to $w_3$.}
  \label{fig:twoCycles}
\end{figure}

The solution to the problem of interfering concurrent integrations is
to serialize all integration steps. For this purpose node $v_0$ acts
as a coordinator. In each of the final phases each node $v$ outside
$C$ first checks if can be integrated using the above described
sequence of messages $I_1$ to $I_3$. If this is the case then $v$
randomly selects one of these possibilities and informs $v_0$. This
message includes information about the four nodes on $C$ that
characterize the integration (see below for details). Node $v_0$
selects among all offers a single node $v$ and informs it.
Upon receiving the integration order, a node $v$ initialize the
integration which is completed after fives rounds. Then the
integration of the next node can start.

The solution for the second problem -- the reversal of the segment --
is based on an ascending numbering of the nodes. Such a numbering can
easily be established in the first and middle phases. During phases 0
and 1 the nodes are numbered as follows: Node $v_0$ has number $0$. In
clockwise order the nodes have numbers $n^{14}, 2n^{14}, 3n^{14},
\ldots$, $\beta n^{14}$ for some integer constant $\beta\le \lceil4
\log n\rceil$. Thus, the difference between two consecutive nodes is
$n^{14}$. During the middle phases when a node $v$ is integrated into
$C$ between two nodes with numbers $f<l$ the integrated node gets the
number $\lceil(f+l)/2\rceil$. This is an integer strictly between $f$
and $l$ as long as $|f-l|\ge 2$. If a node is integrated between $v_0$
and the node with the highest number $y$, the new number is
$y+\lceil(\beta +1)n/2\rceil$. It is straightforward to verify that
all numbers are different and are ascending along the cycle beginning
with $v_0$. The choice of the initial numbers guarantees that the
difference of the numbers of two consecutive nodes is always at least
$2$.

In case a node $v$ is integrated during the final phase it gets the
number $\lceil(n_1+n_2)/2\rceil$ as if it would be inserted between
$w_1$ and $w_2$ with numbers $n_1$ and $n_2$ (see Fig.~\ref{fig:fph}).
The numbers of the nodes between $w_2$ and $w_4$ need to be updated
such that overall the numbers are ascending. When a node can be
integrated it includes in the notification message to $v_0$ the
numbers of the end nodes of the segment that would be reversed if this
node is integrated, i.e., the numbers of $w_2$ and $w_4$ (referred to
as $f$ and $l$ in the following). Afterwards, when $v_0$ informs the
selected node it distributes a message to all nodes in the network
that also includes the numbers $f$ and $l$. A node receiving this
message checks if its own number $x$ is between $f$ and $l$. In this
case it changes its number to $f+l-x$. Thus, the numbers of the nodes
in the segment are reflected on the mid point of the segment (see
Fig.~\ref{fig:rings}). Each node that changes its number also updates
it next pointer to the other neighbor on $C$. Also nodes $v,w_1$, and
$w_2$ update their next pointer.

\begin{figure}[h]
  \centering
\includegraphics[scale=1.0]{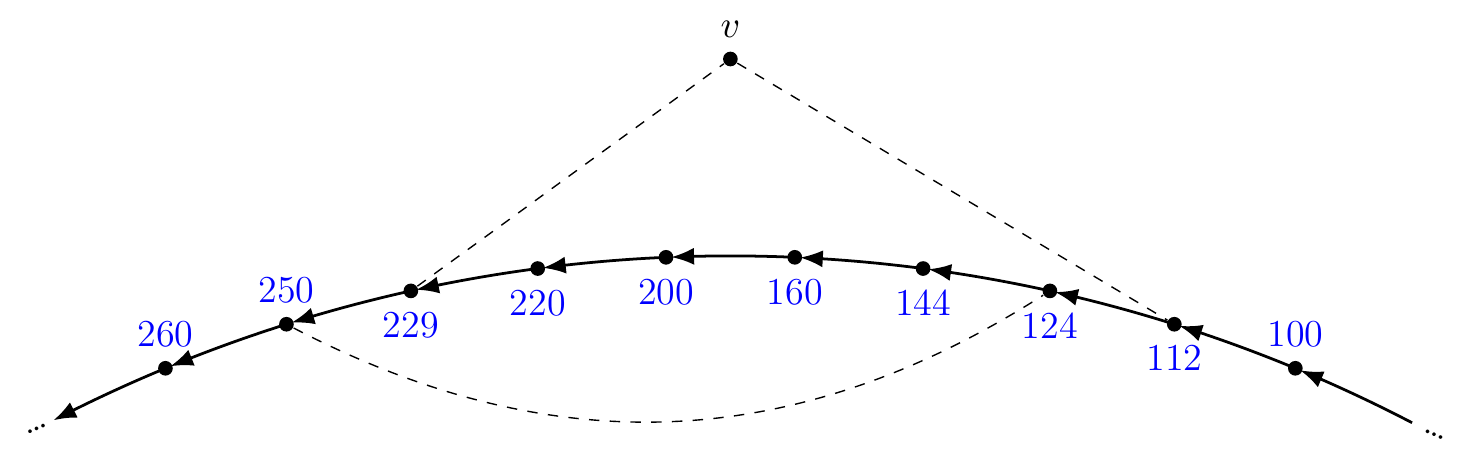}  %
\includegraphics[scale=1.0]{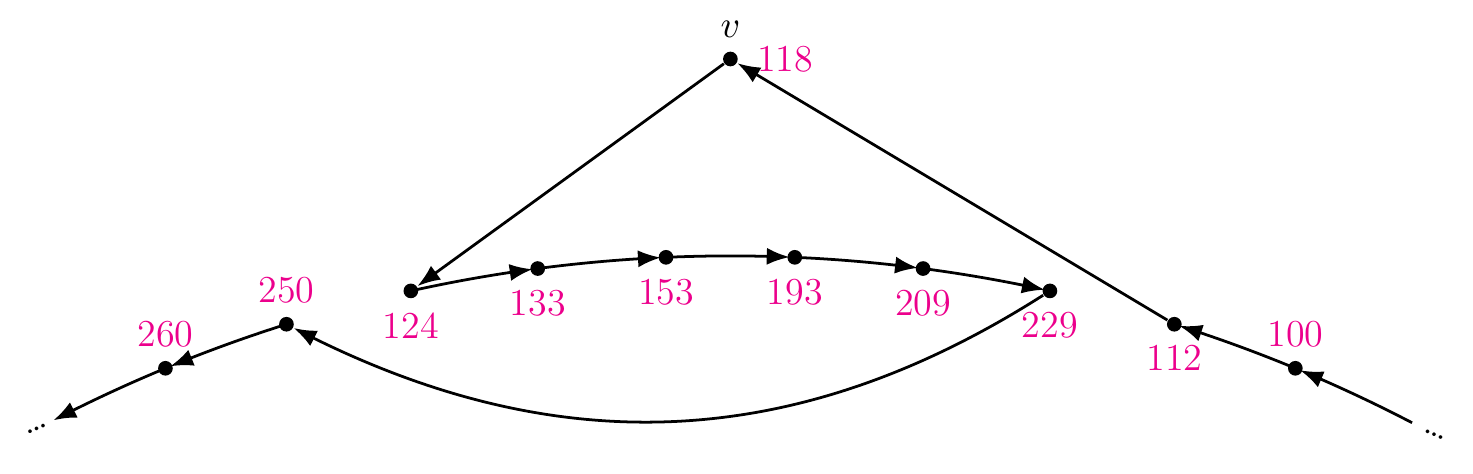}  
\caption{Node $v$ is to be integrated into $C$. The nodes $w_1$ and
  $w_2$ have the numbers $f=124$ and $l=229$. Node $v$ will receive
  number $\lceil(112+124)/2\rceil= 118$. Upon receiving the message
  form node $v_0$, nodes with a number between 124 and 229 change
  their numbers. The left sides shows the old numbers and the right
  side the new numbers.}
  \label{fig:rings}
\end{figure}

This procedure results in a cycle including $v$ with a numbering that
is consistent with the orientation. Thus, when the integration phase
of the next node starts, cycle $C$ is in a consistent state. To carry
out this phase a short route from each node to $v_0$ and vice versa is
needed. This is provided by the BFS tree constructed in the
pre-processing phase: Each node reaches $v_0$ in at most 3 hops. Thus,
each final phase lasts 11 rounds.

\section{Analysis of Algorithm \AlgoHC}
\label{sec:analysis}

This section proves the correctness and analyzes the complexity of the
individual phases and proves the main theorem. First, we prove that
\AlgoHC produces the numbering that guarantees that the final phases
work correctly. Afterwards the individual phases are analyzed. Some of
the results are proved for values of $p$ less than $(\log
n)^{3/2}/\sqrt{n}$ to make them more general.

\begin{lemma}\label{lem:numb}
  At the end of each phase each node has a different number and the
  numbers are ascending beginning with number $0$ for node $v_0$ in
  clockwise order.
\end{lemma}
\begin{proof}
  After phase 1 starting with node $v_0$ the nodes have the numbers
  $n^{14}, 2n^{14}, 3n^{14}, \ldots$, $\beta n^{14}$, i.e., the
  difference between the numbers of two neighboring nodes on $C$ is
  $n^{14}$. A node $v$ that is inserted between two nodes with
  integral numbers $x$ and $y$ in middle phase gets the number
  $\lceil(x+y)/2\rceil$. Let $x<y$. If $x+y$ is even then
  $|x-\lceil(x+y)/2\rceil| = |y-\lceil(x+y)/2\rceil| =d/2$. If $x+y$
  is odd then $|x-\lceil(x+y)/2\rceil| =(d+1)/2$ and
  $|y-\lceil(x+y)/2\rceil| =(d-1)/2$. This yields that the distance
  $d$ between two consecutive numbers is approximately at most cut in
  half, i.e., the smaller part is at least $(d-1)/2$. After $i$ middle
  phases the distance between to numbers is at least

  \begin{equation}
    \label{eq:dist}
\frac{d}{2^i}-\left(1-\frac{1}{2^i}\right)
\end{equation}

  Since there are $16\log n$ middle phases the distance between two
  numbers is $n^{14}/2^{16\log n} -(1 - 1/2^{16\log n})> 2^{(3\log n)+1}$.
  This implies that after the middle phases the numbering of the nodes
  satisfies the stated condition.

  Let $v$ be a node that is inserted in a final phase into $C$. Assume
  that the smallest distance between the numbers of two consecutive
  nodes on $C$ is at least $2$. Consider Fig.~\ref{fig:fph} for
  reference. Let $f$ (resp.\ $l$) the number of $w_1$ (resp.\ $w_3$)
  at the beginning of the corresponding final phase. Denote the nodes
  between $w_2$ and $w_4$ by $w'_1, \ldots, w'_k$ with $w_2=w'_1$ and
  $w_4=w'_k$. Furthermore, let $n'_1,\dots, n'_k$ be the numbers of
  these nodes. Thus, \[f < n'_1 < \ldots < n'_k < l\] The order of
  these nodes on $C$ at the end of the phase will be $w_1, v, w'_k,
  \ldots, w'_1, w_3$. Denote by $n_i$ the new number of node $w'_i$,
  i.e., $n_i=n'_1+n'_k-n'_i$. Thus, we need to prove \[f < \lceil(f+
  n'_1)/2\rceil < n_k < n_{k-1} < \ldots < n_1< l\] Since $n'_1 > f+1$
  it follows $f < \lceil(f+ n'_1)/2\rceil$ and since $n_k=
  n'_1+n'_k-n'_k=n'_1>f+1$ it follows $\lceil(f+ n'_1)/2\rceil < n_k$.
  Furthermore, $n'_i< n'_{i+1}$ implies $n_{i+1} = n'_1+n'_k-n'_{i+1}<
  n'_1+n'_k-n'_{i}=n_{i}$. Finally, $n_1 = n'_1+n'_k-n'_1 = n'_k< l$.
  
  As shown above at the end of the middle phases $d>2^{(3\log n) +1}$.
  Hence, after the last of the $3\log n$ final phases we have $d> 1$
  by equation~(\ref{eq:dist}). Thus, the numbers of all nodes are
  different and ascending.
\end{proof}

The challenge in proving properties of iterative algorithms on random
graphs is to organize the proof such that one only slowly uncovers the
random choices in the input graph while constructing the desired
structure, e.g., a Hamiltonian cycle. This is done in order to cleanly
preserve the needed randomness and independence of events that
establish the correctness proof. The {\em coupling technique} is well
know to solve this problem (\cite{Frieze:15}, p. 5). For
$\gamma \in \mathbb{N}$ let
$\hat{p} = 1- (1 -p)^{{1}/{\gamma\log n}}$. Then
$p = 1-(1-\hat{p})^{\gamma\log n}$. Thus $G(n,p)$ is equal to the
union of $\gamma\log n$ independent copies of $G(n,\hat{p})$. For
$p = (\log n)^{3/2}/\sqrt{n}$ we
have 
\[\left(1 - \frac{\sqrt{\log n}}{\gamma\sqrt{n}}\right)^{\gamma\log n} =  e^{((\log n)^{3/2}/\sqrt{n}) \;\log(1- \sqrt{\log n}/\gamma\sqrt{n}) \gamma\sqrt{n}/\sqrt{\log n}}\ge e^{-(\log n)^{3/2}/\sqrt{n}} \ge 1 - \frac{(\log n)^{3/2}}{\sqrt{n}}\]
hence $\hat{p} \ge \sqrt{\log n}/\gamma\sqrt{n}$ and thus,
\[ \bigcup_{i = 1}^{\gamma\log n} G(n,\sqrt{\log n}/\gamma\sqrt{n})
  \subseteq G(n,p).\] We superimpose $\gamma\log n$ independent copies
of $G(n,\sqrt{\log n}/\gamma\sqrt{n})$ and replace any double edge
which may appear by a single one. In the following proof in each phase
we will uncover a new copy of $G(n,\sqrt{\log n}/\gamma\sqrt{n})$.
There will be $21\log n$ phases, thus $\gamma = 21$. We set
$q =\sqrt{\log n}/\gamma\sqrt{n}$ for the rest of this paper. All but
the final phases also work for values of $p$ slightly smaller than
$(\log n)^{3/2}/\sqrt{n}$ and thus smaller values of $q$ (i.e.,
$q=1/\gamma\sqrt{n}$ for $p=\log n/\sqrt{n}$). This is reflected in
the following proofs.

For $i\ge 0$ let $G^i$ be the union of $i$ independent copies of
$G(n,q)$. In phase $i$ the constructed cycle $C$ consists of edges
belonging to $G^i$. The subsequent proofs use the following fact: The
probability that any two nodes of $V$ are connected with an edge from
$G^{i+1}\setminus G^i$ is $q$. Thus, in each phase a new copy of
$G(n,q)$ is revealed. In each phase we consider the nodes outside $C$.
For each such node we consider {\em unused} edges incident to it, each
of those exist with probability $q$ independent of the choice of $C$,
because $C$ consist of edges of other copies of $G(n,q)$. Some of
these unused edges may also exist in $G^i$, but that does not matter.


\subsection{Phase 0}
Phase 0 sequentially builds a path $P$ by randomly choosing a node to
extend $P$. Even for $p=\log n/ n$ this allows to build paths of
length $\Omega(\sqrt{n})$ in time proportional to the length of $P$.
Since we aim at a runtime of $O(\log n)$ the following lemma
suffices to prove that \whp phase 0 terminates successfully.

\begin{lemma}\label{lem:ph0} If $q\ge \log n/\gamma n$
   phase 0 completes \whp after $3\log n$ rounds
  with a path of length $3\log n$.
\end{lemma}
\begin{proof}
  The probability that an end node of $P$ does not receive a response
  is equal to at least $(1-q)^{n-c}$, where $c$ is the number of nodes already
  in $P$. Thus, the probability to find a path of length $3\log n$ is
  \[\prod_{x=1}^{3\log n}(1 - (1-q)^{n-x}) \ge (1 - (1-q)^{n-3\log
      n})^{3\log n}.\] By Lemma~\ref{lem:levy} (see Appendix)
  $\lim_{n \rightarrow \infty} (1 - (1-q)^{n-3\log n})^{3\log n}=1$,
  this proves the lemma.
\end{proof}

\subsection{Phase 1}
Phase 1 sequentially tries to extend $P$ into a cycle $C$ in at most
$3\log n$ rounds.
\begin{lemma}\label{lem:ph1}
  If $q\ge 1/\gamma\sqrt{n}$ phase 1 finds \whp in $3\log n$ rounds a cycle
  with at most $4\log n$ nodes.
\end{lemma}
\begin{proof} 
  By considering only the edges of the fresh copy of $G(n,q)$ we note
  that the probability that path $P$ cannot be closed into a cycle
  within $3\log n$ rounds is at most
\[\prod_{i=0}^{\log n -1} (1-q^2)^{n-3\log n -i}= (1-\frac{1}{n})^S\]
with \[S=\sum_{i=0}^{\log n - 1}n-3\log n -i=\log n\left(n-\log n -
(\log n -1)/2)\right).\]
By Lemma~\ref{limits} $(1-1/n)^S$ approaches 0 as $n$ goes to
infinity. This completes the proof. 
\end{proof}

\subsection{Middle Phases}
The middle phases contribute the bulk of nodes towards a Hamiltonian
cycle. In each phase the number of nodes is increased by a constant
factor \whp by concurrently testing all edges in $C$ for an extension.
In the following we prove a lower bound for the number of nodes that
are integrated \whp into $C$ in a middle phase. This will be done in
two steps. First we state a lower bound for the number of nodes $v\in
V\setminus C$ that send an invitation $I_2$. Based on this bound we
prove a lower bound for the number of nodes that received an
acceptance message $I_3$. Note that each node $v\in V\setminus C$ that
receives an acceptance message $I_3$ is integrated into $C$ and each
$v\in V\setminus C$ receives at most one $I_3$ message.

Let $c=|C|$ and $v\in V\setminus C$. The event that an edge $e$ of $C$
together with $v$ forms a triangle has probability $q^2$.
Unfortunately these events are not independent in case the edges have
a node in common. To have a lower bound for the probability that $v$
is connected to at least one pair of consecutive nodes on $C$ we
consider only every second edge on $C$. Denote the edges of $C$ by
$e_0,\ldots, e_{c-1}$ with $e_i = (v^i,u^i)$. Let $\pi_{v,i}$ be the
event that node $v$ forms a triangle with edge $e_{2i}$ such that the
edges $(v,v^i)$ and $(v,u^i)$ belong to newly uncovered copy of
$G(n,q)$. For fixed $v$ the events $\pi_{v,i}$ are independent and
each occurs with probability $q^2$. Let $\pi_{v}$ be the event that
for node $v\in V\setminus C$ at least one of the events
$\pi_{v,0}, \pi_{v,2},\pi_{v,4},\dots,\pi_{v,c}$ occurs. Clearly the
events $\pi_v$ are independent and each occurs with probability
$1-(1-q^2)^{c/2}$.

For $v\in V\setminus C$ let $X_v$ be a random variable that is $1$ if
event $\pi_v$ occurs. The variables $X_{v_1},\ldots, X_{v_{n-c}}$ are
independent Bernoulli-distributed random variables. Define a random
variable $X$ as
\[X = \sum_{v\in V\setminus C} X_v.\] 
Then we have
\begin{equation}
  \label{eq:X}
  E[X] =(n-c)(1-(1-q^2)^{c/2}).
\end{equation}
Obviously $X$ is a lower bound for the number of nodes of
$V\setminus C$ that are connected to at least one pair of consecutive
nodes on $C$, i.e., the number of nodes $v\in V\setminus C$ that sent
an invitation $I_2$.



Next let $Y$ be a random variable denoting the number of nodes of
$V\setminus C$ that receive an acceptance message $I_3$ provided that
$X=x$ nodes sent an invitation $I_2$. We compute the conditional
expected value $E[Y|X=x]$. The computation of $Y$ can be reduced to
the urns and balls model: The number of balls is $x$ and the number of
bins is $c$. Each ball is thrown randomly in any of the $c$ bins. Note
that the probability that a node $v$ in $C$ is connected to a node $w$
in $V\setminus C$ is independent of $v$ and $w$ at least $q$. Thus,
$Y$ is equal to the number of nonempty bins and hence
\begin{equation}
  \label{eq:Y}
  E[Y|X=x] = c\left(1-\left(1-\frac{1}{c}\right)^x\right).
\end{equation}
Note that for a given value of $x$ variable $Y$ is the number of nodes
inserted into $C$ in one phase. $Y/c$ is the ratio of the number of
newly inserted nodes to the number of nodes in $C$. The next
subsections give a lower bound for $Y/c$ that holds \whp We
distinguish the cases $x\ge n/7$ and $x< n/7$. The reason is that the
variance of $X$ behaves differently in these two ranges: For $x< n/7$
the variance is rather large, whereas for $x\ge n/7$ the variance
tends to $0$. In both cases we first compute a lower bound for $X$ and
then derive a lower bound for $Y/c$ with respect to the bound for $X$.

Instead of using $q=\sqrt{\log n}/\gamma\sqrt{n}$ the analysis of the
middle phases is done for the smaller value $q=1/\sqrt{n}$. This saves
us from using the constant $\gamma$ and simplifies the exposition of
the proofs.


\subsection{The case $c< n/7$}
Next we prove that while $c< n/7$ in each middle phase the number of
nodes in $C$ is increased by a factor of $2-e^{-1/3}$ and that after
$3 \log n$ phases the bound $n/7$ is exceeded.

\begin{lemma}\label{lem:lowerCaseXS}
  Let $3\log n < c < n/7$. Then there exists $d > 0$ such that
  $X > c/3$ with probability $1-1/n^d$.
\end{lemma}
\begin{proof}
  From equation~(\ref{eq:X}) and Lemma~\ref{lem:x2} (see Appendix) it follows that 
  \[E[X] =(n-c)(1-(1-q^2)^{c/2})> c/3 > \log n.\] Thus, $c/(3E[X]) < 1$ for
  $3\log n < c < n/7$. Also, $c/(3E[X])$ is strictly monotonically
  increasing in this range for fixed $n$. Furthermore, for fixed $n$
  we have
  \[\lim_{c \rightarrow n/7} \frac{c}{3E[X]}= \frac{1}{18(1-e^{-1/14})}< 0.81.\] Thus, for $c$ in the specified range \[\lim_{n \rightarrow \infty} \left(1 - c/(3E[X])\right)^2 > 0.037.\]
  Let $\delta = 1 - c/(3E[X])$. Then $0< \delta < 1$ and we have
  \[E[X]\delta^2 = E[X]\left(1 - c/(3E[X])\right)^2\ge 
    0.037 \log n\]
  for $3\log n < c < n/5$. Hence,
  $e^{-E[X]\delta^2/2}\le 1/n^{0.037/2}$. The Chernoff
  bound (Lemma~\ref{chernoff}) yields that
\[X > (1-\delta) E[X] = \left(1 - 1 + \frac{c}{3E[X]}\right)E[X] =
  c/3\]
 with probability at least $1-1/n^{0.037/2}$.
\end{proof}

\begin{lemma}\label{lem:lowerCaseYS}
  Let $\beta=0.92$ and $3\log n < c < n/7$.
  Then there exist $d > 0$ such that
  $\frac{Y}{c} \ge \beta\left(1-\frac{1}{e^{1/3}}\right)$ with
  probability $1-1/n^d$.
\end{lemma}
\begin{proof}
  From  equation~(\ref{eq:Y})  %
  it follows
\[E[Y|X\ge c/3] \ge c\left(1-(1-\frac{1}{c})^{c/3}\right).\] Let
$\delta^2 = 3\alpha\log n/c$ with $\alpha=(1 -\beta)^2$. Then
$\delta^2 <1$ and
  \[e^{-E[Y|X\ge c/2]\delta^2/3}\le e^{-3\alpha \log n (1-(1-1/c)^{c/3})/2} = \left(\frac{1}{n}\right)^{3\alpha(1-(1-1/c)^{c/3})/2}.\] 
The Chernoff bound (Lemma~\ref{chernoff}) implies that 
\[Y|(X\ge c/3) > (1-\delta) E[Y|X\ge c/3] \ge \left(1 -
  \sqrt{\frac{3\alpha\log n}{c}}\right)c(1-(1-\frac{1}{c})^{c/3})\]
with probability $1 - 1/n^{3\alpha(1-(1-1/c)^{c/3})/2}$. Hence, by
Lemma~\ref{lem:lowerCaseXS} there exists $d > 0$ such that 
\[Y \ge \left(1 -
  \sqrt{\frac{3\alpha\log n}{c}}\right)c(1-(1-\frac{1}{c})^{c/3})\]
with probability $1-1/n^d$.
This gives for any $c\ge 3 \log n$
\[\frac{Y}{c} =\left(1 - \sqrt{\frac{3\alpha\log n}{c}}\right)(1-(1-\frac{1}{c})^{c/3})\ge \beta(1 - \frac{1}{e^{1/3}}).\]
\end{proof}

\begin{lemma}\label{lem:complLS}
  Let $C$ be a cycle with at least $3\log n$ nodes. Then after at most
  $3 \log n$ phases $C$ has \whp at least $n/7$ nodes.
\end{lemma}
\begin{proof}
  Lemma~\ref{lem:lowerCaseYS} yields that while the circle has less
  than $n/7$ nodes \whp in $i$ phases the number of nodes in $C$ grows
  from $c$ to $(1 +\beta(1-\frac{1}{e^{1/3}}))^ic$, i.e., in three
  phases to $(1 +0.92(1-\frac{1}{e^{1/3}}))^3c > 2 c$, i.e., it
  doubles at least every three phases. Hence, starting with
  $c=3\log n$, after at $i$ phases $C$ has at least $2^{i/3}3\log n$
  nodes. Note that $2^{i_0/3}3\log n \ge n/7$ for
  $i_0=3\log \left(n/(21\log n)\right)/\log 2$. Since
  $3\log n \ge i_0$, the union bound implies that after at most
  $3\log n$ phases \whp the circle has at least $n/7$ nodes.
\end{proof}

\subsection{The case $c\ge n/7$}
Next we show that the size of $C$ is still growing by a constant
factor, but the factor is decreasing in each phase. This allows to
infer that after $13\log n$ phases \whp $C$ has at least $n -3\log n$
nodes. Let $c = \xi n$ and \[\Xi=\left(1 - \sqrt{\frac{3\log
        n}{n(1-\xi)}}\right)c(1/\xi-1)(1-(1-q^2)^{c/2}).\] 

\begin{lemma}\label{lem:upperCaseXS}
  Let $c = \xi n$ with $1/7 \le \xi < 1 - 3(\log n)/n$. Then there
  exists $d > 0$ such that $X> \Xi$ with probability $1-1/n^d$.
\end{lemma}
\begin{proof}
  Let $\delta^2 = 3\log n/(n-c)$. Then $0 < \delta< 1$. From equation
  (\ref{eq:X}) and the Chernoff bound (Lemma~\ref{chernoff}) it
  follows
  \[Prob[X \le (1-\delta) E[X]]\le e^{- E[X]\delta^2/2} = e^{-3\log n
      (1-(1-q^2)^{c/2})/2} = 1/n^{3(1-(1-q^2)^{c/2})/2}.\] Note that
  $c \ge n/7$ implies $cq^2/2 = \xi/2 \ge 1/14$ and hence
  $(1-q^2)^{c/2} \le e^{-\xi/2}$ by Lemma~\ref{limits}. Thus,
  $Prob[X \le (1-\delta) E[X]]\le 1/n^{3(1-e^{-\xi/2})/2}$. Since
  $(1-\delta) E[X]=\Xi$ this yields the result.
\end{proof}

Note that this Lemma proves that \whp in each phase there exists at
least one node that can be used to extend the cycle as long as
$c < n -3\log n$ holds.

\begin{lemma}\label{lem:upperCaseS}
  Let $c = \xi n$ with $1/7 \le \xi < 1 - 3(\log n)/n$. Then there
  exists $d > 0$ such that
  $\frac{Y}{c} \ge \left(1 - \sqrt{\frac{3\log
        n}{n(1-\xi)}}\right)\left(1-e^{(1 -
      1/\xi)(1-e^{-\xi/2})}\right)$ with probability $1-1/n^d$.
\end{lemma}
\begin{proof} 
  From equation~(\ref{eq:Y}) it follows
\[E[Y|X\ge \Xi] \ge c\left(1-(1-1/c)^{\Xi}\right).\]
Since $\lim_{n
    \rightarrow \infty} (1-q^2)^{c/2}=e^{-\xi/2}$ by Lemma~\ref{limits} it
  follows \[ \lim_{n
    \rightarrow \infty} \frac{\Xi}{c} \ge
  \lim_{n
    \rightarrow \infty} \left(1 - \sqrt{\frac{3\log
        n}{n(1-\xi)}}\right)\left(\frac{1}{\xi}-1\right)(1-e^{-\xi/2}) = \left(\frac{1}{\xi}-1\right)(1-e^{-\xi/2})\] and therefore by Lemma~\ref{limits}
  \[E[Y|X\ge \Xi]/c=(1-(1-1/c)^{\Xi})\ge 1 - e^{(1 - 1/\xi)(1-e^{-\xi/2})}> 0.\] Next
  let $\delta^2 = 3\log n/c$. Then
  \[e^{-E[Y|X\ge \Xi]\delta^2/2}\le e^{-3\log n (1 - e^{(1 - 1/\xi)(1-e^{-\xi/2})})/2} =
  \left(\frac{1}{n}\right)^{3(1 - e^{(1 - 1/\xi)(1-e^{-\xi/2})})/2}.\] Hence, the
  Chernoff bound implies that with probability
  $1-1/n^{3(1 - e^{(1 - 1/\xi)(1-e^{-\xi/2})})/2}$
  \[Y(X\ge\Xi) > (1-\delta) E[Y|X\ge\Xi] \ge \left(1 -
      \sqrt{\frac{3\log n}{c}}\right)c(1-(1-\frac{1}{c})^{\Xi}).\] The
  result follows from Lemma~\ref{lem:upperCaseXS}.
\end{proof}

\begin{lemma}\label{lem:complUS}
  Let $p\ge \log n/\sqrt{n}$ and $C$ be a cycle with at least $n/7$ nodes.
  Then after $13\log n$ phases $C$ has \whp at least $n -3\log n$
  nodes.
\end{lemma}
\begin{proof}
  If $c= \xi n\ge n/7$ then by Lemma~\ref{lem:upperCaseS} \whp in one
  phase the number of nodes in $C$ grows from $\xi n$ to
  $(1 + \rho(1-e^{(1- 1/\xi)(1-e^{-\xi/2})}))\xi n$, where
  $\rho$ can be arbitrary close to 1. Thus, \whp the number of nodes
  strictly increase per round, but the increase decreases. For example
  the size of $C$ grows in three rounds from $n/7$ to $0.191
  n$ to $0.252 n$. Let $f(c)
  =(2-e^{(1-1/c)(1-1/e^{c/2})})c$. Note that $f(c)\ge a x +
  b$ for $c\in (1/5,1)$, $a=(1-f(1/5))/0.8\approx 0.92085$, and $b=1-a
  \approx 0.07914$. Let $c_0=1/5$ and $c_{i}= (2-e^{(1- 1/c_{i-1})(1-
    1/e^{c_{i-1}/2)}})c_{i-1}$. Thus,
  \[c_i \ge a c_{i-1} + b \ge a^2 c_{i-2} + ab + b\ge \ldots \ge a^i
  (c_0-1) + 1. \] Lemma~\ref{lem:upperCaseS} yields that after another
  $i$ rounds $C$ contains at least $c_i\,n\ge \left(a^i (c_0-1) +
    1\right)n$ nodes. Let $\alpha$ such $1+\alpha\log a < 0$, e.g.,
  $\alpha =13$. Hence, for larger values of $n$ we have $c_{\alpha\log
    n} \ge 1-\frac{4}{5}a^{\alpha \log n} = 1 - \frac{4}{5} n^{\alpha
    \log a} \ge 1 - \frac{3\log n}{n}$. Thus, $c_{13\log n}\,n \ge n -
  3\log n.$ The lemma follows from the union bound.
\end{proof}


\subsection{Final Phases}
After the middle phases \whp there are at most $3\log n$ nodes outside
$C$. The following lemma proves the correctness of the final phases. 



\begin{lemma}\label{lem:extInt}
  If $p \ge (\log n)^{3/2}/\sqrt{n}$ the final $3\log n$ phases
  integrate \whp all remaining nodes into $C$.
\end{lemma}
\begin{proof} 
  Let $v\in V\setminus C$ be a fixed node. As before, we only consider
  edges incident to $v$ that belong to a fresh copy of $G(n,q)$. Let
  the random variable $X$ denote the number of neighbors of $v$ on
  $C$. If $C$ consists of $c$ nodes then
  $E[X]\ge \frac{c}{\gamma}\sqrt{\log n/n}$. Let
  $\delta^2 = \frac{2 \sqrt{n}}{c}$. Then $\delta^2 < 1$ and
  $\lim_{n \rightarrow \infty} E[X]\delta^2/2 \ge \lim_{n \rightarrow
    \infty}\sqrt{\log n}/\gamma =\infty$. Now the Chernoff bound implies that
  \whp
\[X \ge \left(1- \sqrt{\frac{2\sqrt{n}}{c}}\right)\frac{(n-3\log
    n)\sqrt{\log n}}{\gamma\sqrt{n}}.\]
For $i=n-3\log n, \ldots n-1$ let
\[X_i=\left(1- \sqrt{\frac{2\sqrt{n}}{i}}\right)\frac{(n-3\log
    n)\sqrt{\log n}}{\gamma\sqrt{n}}\]
Now, by the union bound, the probability that the final phases do not
integrate all remaining $3\log n$ nodes is at most
\[\sum_{i=n-3\log n}^{n-1} \left(1 -
    \frac{1}{\gamma}\sqrt{\frac{\log n}{n}}\right)^{\frac{X_i(X_{i}+1)\sqrt{\log n}}{2\sqrt{n}}} \le 3\log n \left(1 -
    \frac{1}{\gamma}\sqrt{\frac{\log n}{n}}\right)^{\frac{(X_{n-3\log n})^2\sqrt{\log n}}{2\sqrt{n}}}.\]
Lemma~\ref{lem:finalLimit} (see Appendix) shows that this term converges to $0$.
\end{proof}

\section{Proof of Theorem~\ref{thm:main-contribution}}
The pre-processing phase lasts $9$ rounds. By Lemma~\ref{lem:ph0} and
\ref{lem:ph1} phases 0 and 1 terminate after $O(\log n)$ rounds \whp
with a cycle with at most $4\log n$ nodes. Each middle phase lasts a
constant number of rounds. According to Lemma~\ref{lem:complLS} after
at $3\log n$ middle phases the cycle $C$ has \whp $n/7$ nodes and by
Lemma~\ref{lem:complUS} after another $13 \log n$ middle phases \whp
$n-3\log n$ nodes. Then in $3\log n$ final phases, each lasting a
constant number of rounds, $C$ is \whp a Hamiltonian cycle by
Lemma~\ref{lem:extInt}. This leads to the total time complexity of
$O(\log n)$ rounds. The statements about message size and memory per
node are evident from the description of \AlgoHC.


\section{Conclusion}
\label{sec:conclusion}
This paper presented an efficient distributed algorithm to compute in
$O(\log n)$ rounds \whp a Hamiltonian cycle for a random graph
$G(n,p)$ provided $p\ge (\log n)^{3/2}/\sqrt{n}$. This constitutes a
large improvement over the state of the art with respect to
$p =c \log n/n^\delta$ ($0<\delta \le 1$) and run time
$\tilde{O}(n^\delta)$. It is well known that $G(n,p)$ contains \whp a
Hamiltonian cycle, provided $p \ge p_{crit}$. There is a large gap
between $(\log n)^{3/2}/\sqrt{n}$ and $p_{crit}$. It appears that by
maxing out the arguments of this paper it is possible to prove
Theorem~\ref{thm:main-contribution} for $p=\sqrt{\log n/n}$. All but
the final phases already work for $p\ge \log n/\sqrt{n}$. We suspect
that finding a distributed $O(\log n)$ round algorithm for
$p\in o(1/\sqrt{n})$ is a hard task.






\section{Acknowledgments}
This work is supported by the Deutsche Forschungsgemeinschaft (DFG)
under grant DFG TU 221/6-2. The author is grateful to the
reviewers' valuable comments that improved the manuscript.

\bibliographystyle{abbrv}
\bibliography{arxiv_ham}  


\appendix

The appendix is divided in two sections. The first section contains
technical results which did not fit into the paper due to space
restrictions. The second section contains well known results without
stating a proof, these are included to make the paper
self-contained.

\section{Technical Lemmas}
\label{sec:technical-lemmas}

\begin{lemma}\label{lem:x2}
  There exists $n_0>0$ such that $(n-c)(1-(1-\frac{1}{n})^{c/2})> c/3$
  for all $0 < c < n/7$ and $n\ge n_0$.
\end{lemma}
\begin{proof}
  Obviously it suffices to prove
  \begin{equation}
    \label{eq:2x}
    \log\left(\frac{3n-4c}{3(n-c)}\right)> \frac{c}{2}
    \log\left(1-\frac{1}{n}\right).    
  \end{equation}
  The derivative of the left side (considering $n$ as a constant)
  is \[\frac{-n}{(3n-4c)(n-c)}\] This is larger than
  $(1/2)\log\left(1-\frac{1}{n}\right)$, the derivative of the right
  side of equation~(\ref{eq:2x}), in the range $(0,d]$ for some $d>0$.
  Then, at least until the derivatives of both sides are equal,
  equation~(\ref{eq:2x}) is satisfied. The solution $c_0$ of the
  equation
  \[\frac{-n}{(3n-4c)(n-c)} = \frac{1}{2}\log\left(1-\frac{1}{n}\right)\]
  is
  \[c_0= \frac{n}{8}\left(7 - \sqrt{1 - \frac{32}{n\log\left( 1 -
          \frac{1}{n}\right)}}\right).\] Using the rule of
  L'H\^{o}pital we have $\lim_{n \rightarrow \infty} n\log\left( 1 -
    \frac{1}{n}\right) = -1$ This implies that for growing $n$ the
  value of $c_0$ approaches $n(7-\sqrt{33})/8\ge 0.156 n$. Thus for
  some $n$ we have $c_0 \ge n/7$. This proves the lemma.
\end{proof}

\begin{lemma}\label{lem:levy} Let $\alpha, \beta > 1$. For $q \ge \log n /\gamma n$
\[ \lim_{n \rightarrow \infty} (1 - (1-q)^{n-\beta\log
  n})^{\alpha\log n} =1.\]
\end{lemma}
\begin{proof}
Since $\lim_{n \rightarrow \infty} (1-q)^{n-\beta\log n}=0$, hence by
Lemma~\ref{limits} we have
\[\lim_{n \rightarrow \infty} (1 - (1-q)^{n-\beta\log
  n})^{\alpha\log n} = \lim_{n \rightarrow \infty} e^{\alpha\log n
  \log (1 - (1-q)^{n-\beta\log n})} = \lim_{n \rightarrow \infty}
e^{{-\alpha\log n}{(1-q)^{n-\beta\log n}}}.\] Thus, it suffices to
prove $\lim_{n \rightarrow \infty} e^{{-\alpha\log
    n}{(1-q)^{n-\beta\log n}}}= 1$. 
$\lim_{n \rightarrow \infty} \log
n/n = 0$ implies \[\lim_{n \rightarrow \infty} \log (\alpha\log n/n) =
-\infty.\] Thus, $\lim_{n \rightarrow \infty}-\log n + \log (\alpha\log
n) = -\infty$ and hence
  \[\lim_{n \rightarrow \infty}-\log n + \frac{\beta\log ^2 n}{n} +\log (\alpha\log n) = -\infty\]
  \[\lim_{n \rightarrow \infty}-(\log n/\gamma n)\left(n- \beta\log n\right)+\log (\alpha\log n) = -\infty\]
  Since $\lim_{n \rightarrow \infty}\log(1-q)/q = -1$ this yields
    \[\lim_{n \rightarrow \infty}\left(n- \beta\log n\right)\log (1-q)+\log (\alpha\log n) = -\infty\]
    \[\lim_{n \rightarrow \infty}\log\left(\alpha\log n(1-q)^{n- \beta\log n}\right) = -\infty\]
    \[\lim_{n \rightarrow \infty}\alpha\log n(1-q)^{n- \beta\log n} = 0\]
     \[\lim_{n \rightarrow \infty} e^{{-\alpha\log n}{(1-q)^{n-\beta\log n}}}= 1.\]
\end{proof}

  
\begin{lemma}\label{lem:finalLimit}
  \[ \lim_{n \rightarrow \infty} \log n \left(1 -
    \frac{1}{\gamma}\sqrt{\frac{\log n}{n}}\right)^{\frac{(X_{n-3\log n})^2\sqrt{\log n}}{2\sqrt{n}}}= 0.\]
\end{lemma}
\begin{proof}
Since \[\lim_{n \rightarrow \infty} \left( 1
      -\sqrt{\frac{2\sqrt{n}}{n-3\log n}}\right)^2\frac{(n-3\log
      n)^2}{2\gamma}\left(\frac{\log n}{n}\right)^{1.5} \frac{1}{\sqrt{n}}= \infty\]
we have
\[\log n \left(1 -\frac{1}{\gamma}\sqrt{\frac{\log n}{n}}\right)^{\left( 1
      -\sqrt{\frac{2\sqrt{n}}{n-3\log n}}\right)^2\frac{(n-3\log
      n)^2}{2\gamma}\left(\frac{\log n}{n}\right)^{1.5}} < \log n \left(1 -\frac{1}{\gamma}\sqrt{\frac{\log n}{n}}\right)^{\sqrt{n}}\]
The last term is equal to
\[\log n\, e^{(\sqrt{\log n}/\gamma)\log(1 -\frac{1}{\gamma}\sqrt{\frac{\log n}{n}})/\sqrt{\frac{\log n}{n}}} \le \log n/ e^{(\sqrt{\log n}/\gamma)}.\]
This proves the lemma.  
\end{proof}

\section{Well Known Results} 
\label{appd:proof-number-of-heads}

\begin{lemma}\label{limits}
  Let $f_n$ and $p_n$ be sequences with $\lim_{n \rightarrow \infty} p_n = 0$.
  \begin{enumerate}
  \item $\lim_{n \rightarrow \infty}\log(1-p_n)/p_n = -1$.
  \item If $\lim_{n \rightarrow \infty} p_nf_n = c$ then $\lim_{n \rightarrow \infty} (1-p_n)^{f_n}= e^{-c}$.
  \item If $\lim_{n \rightarrow \infty} p_nf_n = \infty$ then $\lim_{n \rightarrow \infty} (1-p_n)^{f_n}= 0$.
  \end{enumerate}
\end{lemma}


\begin{lemma}[Chernoff Bound]\label{chernoff}
  Let $X_1,\ldots, X_n$ be independent Bernoulli-distributed random
  variables and $X= \sum_{i=1}^n X_i$ with $\mu=E[X]$. Then for all
  $0< \delta \le 1$
\[Prob[X \le (1-\delta)\mu]\le e^{-\mu\delta^2/2}.\]
\end{lemma}

\begin{lemma}\label{diameter}
  Let $G(n,p)$ with $p \ge \sqrt{1/n}$. Then \whp $diam(G)\le3$.
\end{lemma}
\begin{proof}
  According to Corollary 8 (i) of \cite{Boll:81} \whp $diam(G)=3$ if
  \begin{itemize}
  \item $({\log n})/3 - 3\log \log n$ converges to $\infty$
  \item $p^3n^2 -2\log n$ converges to $\infty$
  \item $p^2n - 2 \log n$ converges to $-\infty$
  \end{itemize}
This is satisfied for $p = \sqrt{1/n}$.
\end{proof}



\end{document}